\documentclass[twoside,a4paper]{article} 
\usepackage{amsmath,graphicx,amssymb,fancyhdr,amsthm}  
\newtheorem{thm}{Theorem}[section]

\newtheorem{prop}[thm]{Proposition}  
\newtheorem{con}[thm]{Conjecture}  
\theoremstyle{definition}  
  
\theoremstyle{remark}  
  
\def\beq{\begin{eqnarray}}  
\def\eeq{\end{eqnarray}}  
\def\bsp{\begin{split}}  
\def\esp{\end{split}}

\def\d{\mathrm{d}}

\newcommand{\mbold}[1]{\mbox{\boldmath{\ensuremath{#1}}}}

\begin{document}   
   
\title{\Large\textbf{Universality and constant scalar curvature invariants}}   
\author{{\large\textbf{A A Coley$^{1}$ and S Hervik$^{2}$} }   
 \vspace{0.3cm} \\   
$^{1}$Department of Mathematics and Statistics,\\   
Dalhousie University,    
Halifax, Nova Scotia,\\    
Canada B3H 3J5    
\vspace{0.2cm}\\ 
$^2$Faculty of Science and Technology,\\   
 University of Stavanger,\\  N-4036 Stavanger, Norway    
\vspace{0.3cm} \\    
\texttt{aac@mathstat.dal.ca, sigbjorn.hervik@uis.no} }   
\date{\today}   
\maketitle   
\pagestyle{fancy}   
\fancyhead{} 
\fancyhead[EC]{A. Coley, S. Hervik}   
\fancyhead[EL,OR]{\thepage}   
\fancyhead[OC]{Universality}   
\fancyfoot{} 
   
\begin{abstract}

A classical solution is called universal  if  
the quantum correction is a multiple of the metric. 
Universal solutions consequently play an important role in the quantum theory.  
We show that in a spacetime which is universal all of the  scalar curvature invariants 
are constant (i.e., the spacetime is $CSI$).    
   
\end{abstract}   

\section{Universality}

In \cite{CGHP}  
metrics of holonomy $\mathrm{Sim}(n-2)$ were investigated, and it  
was found that all 4-dimensional $\mathrm{Sim}(2)$ metrics  
(which belong to the subclass of Kundt-$CSI$ spacetimes \cite{coley}) are  
universal and consequently can be interpreted as metrics with  
vanishing quantum corrections and are automatically solutions to  
the quantum theory.  

A classical solution is called {\it universal}  if  
the quantum correction is a multiple of the metric, and 
therefore plays an important role in the quantum theory regardless  
of what the exact form of this theory might be.  
That is, if the spacetime is universal, then every symmetric conserved rank-2 tensor, $T_{ab}$,  
which is constructed from the metric, Riemann tensor and its covariant  
derivatives, is of the form  
\beq   
T_{ab} = \mu g_{ab},  
\eeq  
where $\mu$ is a constant. 
Now, for every scalar $S$ that  
appears in the action (gravitational Lagrangian) we obtain by variation 
(since these geometric tensors 
are automatically conserved due to 
the invariance of the actions under spacetime diffeomorphisms) a  
symmetric conserved rank-2 tensor 
$S_{ab}$. For each such tensor we have from the condition of universality 
that $S_{ab} = {\hat{\mu}} g_{ab}$. 
By using an appropriate set of such scalars, we shall 
show that all of the  
scalar curvature invariants 
are constant and that the 
resulting spacetimes are therefore CSI (by definition).
Since the resulting spacetime is  
automatically an Einstein space, in effect 
we must show that  
all scalar contractions of the Weyl tensor and its derivatives 
are constants \cite{CSI,CSI4d}. We utilize the results of FKWC 
\cite{FKWC1992,DecaniniFolacci2007} 
to obtain all conserved rank-2 tensors obtained from variations (by Noether's theorem) of an 
elemental scalar Riemann polynomial. 
 
There are a number of related results we would like to investigate in this paper. 
We will state these in terms of a conjecture and will corroborate this conjecture by proving a 
number of sub-results using a number of different arguments.

\begin{con}\label{con}
A Universal n-dimensional Lorentzian spacetime, $(M,g)$, has the following properties:
\begin{enumerate}
\item{} It is CSI.
\item{} It is a degenerate Kundt spacetime.
\item{} There exists a spacetime, $(\tilde{M},\tilde{g})$, of Riemann type {\bf {D}} having identical scalar polynomial invariants; consequently $(\tilde{M},\tilde{g})$ is spacetime homogeneous.
\item{} There exists a homogeneous isotropy-irreducible Riemannian spacetime $(\hat{M},\hat{g})$ having identical scalar polynomial invariants as  $(M,g)$; i.e.,  $(\hat{M},\hat{g})$ is universal as a Riemannian space.
\end{enumerate} 
\end{con}
 
In low dimensions this conjecture can be proven; in particular,
dimension 2 is trivial as there is only one independent component,
namely the Ricci scalar R.  In dimension 3, there are only Ricci
invariants and the conjecture can be proven by brute force using
symmetric conserved tensors.  Most of our investigation will focus on
dimension 4 and, unless stated otherwise, hereafter we will assume that the
manifold is 4 dimensional.

\section{The CSI result}

Let us first present results substantiating the claim that
universal spacetimes are CSI.  This is clearly the case in the
Riemannian case where Bleecker \cite{Bleecker} showed that the
critical manifolds are homogeneous and, hence, CSI. 
Note that in the Riemannian case a CSI space is equivalent to a
locally homogeneous space; however, in the Lorentzian case these are not equivalent
as there are many examples of CSI spacetimes not being locally homogeneous.

\subsection{The direct method}  
Field theoretic calculations on curved spacetimes are non-trivial due to the 
systematic occurrence, in the expressions involved, of Riemann 
polynomials. These polynomials are formed from the Riemann tensor by 
covariant differentiation, multiplication and contraction. The results of these 
calculations are complicated because of the non-uniqueness of their 
final forms, since the symmetries of the Riemann tensor as well as the
Bianchi identities can not be used in a uniform manner and monomials 
formed from the Riemann tensor may be linearly dependent in 
non-trivial ways. In~\cite{FKWC1992}, Fulling, King, Wybourne 
and Cummings (FKWC)  
systematically expanded the Riemann polynomials encountered in 
calculations on standard bases constructed from group theoretical 
considerations. They displayed such bases for scalar 
Riemann polynomials of order eight or less in the derivatives of the 
metric tensor and for tensorial Riemann polynomials of order six or 
less. We adopt the FKWC-notations ${\cal R}^r_{s,q}$ 
and ${\cal R}^r_{\lbrace{\lambda_1 \dots \rbrace}}$ to denote, 
respectively, the space of Riemann polynomials of rank r (number of 
free indices), order s (number of differentiations of the metric 
tensor) and degree q (number of factors $\nabla^p 
R_{\dots}^{\dots}$) and the space of Riemann polynomials of rank r 
spanned by contractions of products of the type 
$\nabla^{\lambda_1}R_{\dots}^{\dots}$ 
\cite{FKWC1992}. 
The geometrical identities utilized to 
eliminate ``spurious" Riemann monomials include: 
(i) the commutation of covariant derivatives, 
(ii) the ``symmetry" properties of the Ricci and the 
Riemann tensors (pair symmetry, antisymmetry, cyclic symmetry), and 
(iii) the Bianchi identity and its consequences 
obtained by contraction of index pairs.

In this paper we 
actually use a slightly modified version of the FKWC-bases 
\cite{DecaniniFolacci2007}, which are independent of the  
dimension of spacetime 
and  provide irreducible 
expressions for all of our results. In addition, 
the results of \cite{DecaniniFolacci2007}  
provide irreducible expressions for the metric variations (i.e., for 
the functional derivatives with respect to the metric tensor) of the 
action terms associated with the 17 basis elements for the scalar 
Riemann polynomials of order six in derivatives of the metric tensor 
(the so-called curvature invariants of order six).

\subsubsection{Riemann polynomials of rank 0 (scalars)}

The most general expression for a scalar of order six or less in 
derivatives of the metric tensor is obtained by expanding it in the 
FKWC-basis for Riemann polynomials of rank 0 and order 6 or less \cite{FKWC1992}. 
 
{\it The sub-basis for Riemann polynomials of rank 0 and order 2} 
consists of a single element: $R$ [${\cal R}^0_{2,1}$].

Choosing $S$ to be the Ricci scalar, $R$, we find that the Einstein tensor 
is conserved and $R_{ab} = \lambda g_{ab}$,  
where $\lambda$ is a 
constant, and the spacetime is necessarily an Einstein space:  
\beq \label{einstein} 
R_{pq}= \lambda g_{pq};~~ R_{pq ; r}=0.  
\eeq  
Every scalar contraction of the Ricci tensor (or its covariant 
derivatives, which are in fact zero), will thus necessarily be constant. 
Every scalar contraction of the Riemann tensor and its derivatives with the  
Ricci tensor or its covariant 
derivatives will be constant.  
For example, for  $S = R_{ab}R^{ab}$ for an Einstein space we have that $S_{ab} =  
2(R_{acbd} - \frac{1}{4}g_{ab}R_{cd})  R^{cd} = {\tilde{\mu}} g_{ab}$  
(where ${\tilde{\mu}} \sim \lambda^2 $). 
Every mixed invariant (containing both 
the Ricci tensor and the Weyl tensor and their derivatives, will 
be constant or can be written entirely as a contraction of scalars 
involving just the Weyl tensor and its derivatives (up to an additive 
constant term). 
 
Thus to prove that the resulting spacetimes are CSI, we must show that  
all scalar contractions of the Weyl tensor and its derivatives 
are constants.  
 
{\it The sub-basis for Riemann polynomials of rank 0 and order 4} 
has 4 elements: $\Box R$ [${\cal R}^0_{4,1}$]: $R^2$,    
$R_{pq} R^{pq}$,  $R_{pqrs}R^{pqrs}$ [${\cal R}^0_{4,2} $].

From (\ref{einstein}) there is only one rank 0/order 4 independent  
scalar, $C^2 \equiv C_{pqrs} C^{pqrs}$. By varying  $S=C^2$, we obtain a  
symmetric conserved rank-2 tensor which depends on quadratic polynomial 
contractions of the Weyl tensor, which by universality is proportional to 
the metric: 
\beq \label{C1} 
C_{almn}C_{b}^{~lmn}+ C_{blmn}C_{a}^{~lmn} = 2 {\hat{\lambda}} g_{ab}. 
\eeq 
Hence we have that:  
\beq \label{C2} 
C^2  = {\hat{\lambda}}. 
\eeq 
 
Indeed, by choosing $S$ to be a polynomial contraction of the Weyl 
tensor alone (higher than quadratic), we find that by varying  $S$ we obtain   
symmetric conserved rank-2 tensors which depends on polynomial 
contractions of the Weyl tensor which by universality are proportional to 
the metric, and hence all zeroth order invariants constructed from the Weyl tensor 
are constant (and the spacetime is said to be CSI$_0$). 
We note that in higher dimensions, all Lovelock tensors are divergence free 
and consequently (by universality) proportional to the metric. However, we shall 
not proceed in this way here.

The most general expression for a gravitational Lagrangian of order 
six in derivatives of the metric tensor is obtained by expanding it 
in the FKWC-basis for {\it Riemann polynomials of order 6 and rank 
0. This sub-basis} consists of the 17 following elements 
\cite{FKWC1992}: $\Box \Box R$  
$[{\cal R}^0_{6,1}]$: $R\Box R$,   $R_{;p q}R^{pq}$, $R_{pq}\Box R^{pq}$,   
$R_{pq ; rs}R^{prqs}$ 
$[{\cal R}^0_{\lbrace{2,0\rbrace}}]$: 
$R_{;p}R^{;p}$,  
$R_{pq;r} R^{pq;r}$,  $R_{pq;r} R^{pr;q}$,  $R_{pqrs;t}R^{pqrs;t}$ 
$[{\cal R}^0_{\lbrace{1,1 \rbrace}}]$: $R^3$,   $RR_{pq} R^{pq}$,    
$R_{pq}R^{p}_{\phantom{p} r}R^{qr}$, $R_{pq}R_{rs}R^{prqs}$, $RR_{pqrs} R^{pqrs}$,  
$R_{pq}R^p_{\phantom{p} rst} R^{qrst}$,  
$R_{pqrs}R^{pquv}$, $R^{rs}_{\phantom{rs} uv}$,  $R_{prqs} R^{p\phantom{u} q}_{\phantom{p} u \phantom{q} v} R^{r u s v}$ 
$[{\cal R}^0_{6,3}]$.

In general, only 10 of these give rise to independent variations. The other 7 
depend on these via total divergences (and Stokes theorem); 
the functional derivatives (i.e., conserved tensors)  
with respect to the metric tensor of the 
7 remaining action terms can then be obtained in a straightforward manner.

In the case of an 
Einstein space satisfying the conditions 
(\ref{einstein}), (\ref{C1}) and (\ref{C2}), 
there are only three independent rank 0/order 6 scalars: 
\begin{equation}\label{scalars} 
(\nabla{C})^2 \equiv C_{pqrs;t}C^{pqrs;t}, 
C_{1}^3 \equiv C_{pqrs}C^{pquv} C^{rs}_{\phantom{rs} uv},   
C_{2}^3 \equiv C_{prqs} C^{p\phantom{u} q}_{\phantom{p} u \phantom{q} v}  
C^{r u s v} 
\end{equation} 
(where, for example,  ($\nabla{C})^2 \equiv R_{pqrs;t}R^{pqrs;t} =  
C_{pqrs;t}C^{pqrs;t}$).

Variations of the last four scalars in the list above 
give rise to 4 independent conserved rank-2 tensors 
(although $RR_{pqrs} R^{pqrs}$ and 
$R_{pq}R^p_{\phantom{p} rst} R^{qrst}$ are equivalent to   
$\lambda {\hat{\lambda}}$, their variations are non-trivial). 
Note that 
$R_{pqrs;t}R^{pqrs;t}$ 
depends on the other 4 scalars via a total divergence (and Stokes theorem).

\subsubsection{Conserved rank 2 tensors of order six}

The functional derivatives of the ten independent action terms 
in the FKWC-basis were expanded in \cite{DecaniniFolacci2007}: 
for an 
Einstein space satisfying the conditions 
(\ref{einstein}), (\ref{C1}) and (\ref{C2}), we obtain the 
following 4 independent explicit irreducible expressions for the metric variations of the action terms 
 constructed from the 17 scalar Riemann monomials of order six:

\begin{eqnarray}\label{FD_06_9} 
  H_{ab}^{(6,3)(7)} &\equiv  \frac{1}{ 
\sqrt{-g}}\frac{\delta}{\delta g^{ab}} \int_{\cal M} d^D x\sqrt{-g} 
~R_{pqrs}R^{pquv} R^{rs}_{\phantom{rs} uv}  \nonumber \\ 
&    = 24\,  R^{p \phantom{(a}; qr}_{\phantom{p } (a} R_{|pqr|  b)} 
-12\, R^p_{\phantom{p} a;q} R_{p b}^{\phantom{p b};q} + 12\, 
R^p_{\phantom{p} a;q} R^q_{\phantom{q} b;p}    \nonumber \\ 
& \quad + 3\,  
R^{pqrs}_{\phantom{pqrs};a} R_{pqrs ; b }  -6\, 
R^{pqr}_{\phantom{pqr}a;s}R_{pqr b}^{\phantom{pqr b};s} 
-6\, R^{pq}R^{rs}_{\phantom{rs} pa}R_{rs q b}   \nonumber \\ 
& \quad +12\, R^{p r q 
s}R^t_{\phantom{t} pq a}R_{t rs b}  
+ \frac{1}{2}g_{ab} [R_{pqrs}R^{pquv} R^{rs}_{\phantom{rs}uv} ], 
\end{eqnarray} 
which implies that (using (\ref{einstein}), (\ref{C1})-(\ref{C2}))

\begin{eqnarray}\label{FD_06_9a} 
3\, C^{pqrs}_{\phantom{pqrs};a} C_{pqrs ; b }   
-6\,C^{pqr}_{\phantom{pqr}a;s}C_{pqr b}^{\phantom{pqr b};s} +12\, C^{p r q 
s}C^t_{\phantom{t} pq a}C_{t rs b}   \nonumber \\ 
+\frac{1}{2} g_{ab} [C_{pqrs}C^{pquv} C^{rs}_{\phantom{rs} 
uv}] = \lambda_1 g_{ab}. 
\end{eqnarray}

In addition,

\begin{eqnarray}\label{FD_06_7} 
 H_{ab}^{(6,3)(5)} &\equiv \frac{1}{ 
\sqrt{-g}}\frac{\delta}{\delta g^{ab}} \int_{\cal M} d^D x\sqrt{-g} 
~RR_{pqrs} R^{pqrs}  \nonumber \\ 
H_{ab}^{(6,3)(6)} &\equiv \frac{1}{ 
\sqrt{-g}}\frac{\delta}{\delta g^{ab}} \int_{\cal M} d^D x\sqrt{-g} 
~R_{pq}R^p_{\phantom{p} rst} R^{qrst }  \nonumber \\ 
H_{ab}^{(6,3)(8)} &\equiv \frac{1}{ 
\sqrt{-g}}\frac{\delta}{\delta g^{ab}} \int_{\cal M} d^D x\sqrt{-g} 
~R_{prqs} R^{p \phantom{u} 
q}_{\phantom{p} u \phantom{q} v} R^{r u s v}  \nonumber \\ 
\end{eqnarray} 
yield (respectively), 
\begin{eqnarray}\label{eq1} 
2\, C^{pqrs} C_{pqrs ; (a b) }  
+ 2\, C^{pqrs}_{\phantom{pqrs};a} C_{pqrs 
; b } + g_{ab} [-2\, C_{pqrs;t}C^{pqrs;t} \nonumber \\  
+ 2\, C_{pqrs}C^{pquv} C^{rs}_{\phantom{rs} uv}   
+ 8\, C_{prqs} C^{p \phantom{u} q}_{\phantom{p} u \phantom{q} v} C^{r u s 
v} ] = \lambda_2 g_{ab}, 
\end{eqnarray} 
 
\begin{eqnarray}\label{eq2} 
\frac{1}{2} C^{pqrs}C_{pqrs ; (a b) }  
+ \frac{1}{2} C^{pqrs}_{\phantom{pqrs};a} C_{pqrs ; b } - 
C^{pqr}_{\phantom{pqr}a;s}C_{pqr b}^{\phantom{pqr 
b};s}  \nonumber \\ 
+ C^{pq rs}C_{pq t a }C_{rs \phantom{t} b}^{\phantom{rs} t} 
+ 4\, C^{p r q s}C^t_{\phantom{t} pq a}C_{t rs b}  - 
C^{pqr}_{\phantom{pqr} s } C_{pqr t}C^{s 
\phantom{a} t}_{\phantom{s} a \phantom{t} b} \nonumber \\ 
+ \frac{1}{4}g_{ab} [- C_{pqrs;t} C^{pqrs;t} 
+ C_{pqrs}C^{pquv} C^{rs}_{\phantom{rs} uv}  
+4 C_{prqs} C^{p \phantom{u} q}_{\phantom{p} u \phantom{q} v} 
C^{r u s v} ]=  \lambda_3 g_{ab}, 
\end{eqnarray} 
 
\begin{eqnarray}\label{eq3} 
-\frac{3}{4} C^{pqrs} C_{pqrs ; (a b) } 
+ \frac{3}{4} C^{pqrs}_{\phantom{pqrs};a} 
C_{pqrs ; b }   
- \frac{3}{2} C^{pq rs}C_{pq 
t a }C_{rs \phantom{t} b}^{\phantom{rs} t} -9\, C^{p r q s}C^t_{\phantom{t} 
pq a}C_{t rs b}   \nonumber \\ 
 + \frac{3}{2}  C^{pqr}_{\phantom{pqr} s } C_{pqr t}C^{s 
\phantom{a} t}_{\phantom{s} a \phantom{t} b} + \frac{1}{2}g_{ab} [C_{prqs} C^{p \phantom{u} 
q}_{\phantom{p} u \phantom{q} v} C^{r u s v} ] = \lambda_4 g_{ab}. 
\end{eqnarray}

Contracting eqns.  (\ref{FD_06_9a}), (\ref{eq1})-(\ref{eq3}), and using 
$C^{pqrs}\Box {C_{pqrs}} = -C_{1}^3 -4 C_{2}^3 + 2\lambda {\hat{\lambda}}$ 
(etc.) \cite{DecaniniFolacci2007}, 
we then obtain 
\begin{eqnarray}\label{results} 
-3(\nabla{C})^2 + 2C_{1}^3 + 12C_{2}^3 = 4\lambda_1,\nonumber \\ 
-3(\nabla{C})^2 + 3C_{1}^3 + 4C_{2}^3 = 2\lambda_2 - 2\lambda{\hat \lambda} , \nonumber \\ 
-3(\nabla{C})^2 + 3C_{1}^3 + 12C_{2}^3 = 8\lambda_3 - 2\lambda{\hat \lambda},\nonumber \\ 
-3(\nabla{C})^2 + 3C_{1}^3  +16C_{2}^3 = -16\lambda_4  - 6\lambda{\hat \lambda},  
\end{eqnarray} 
and hence the 3 independent scalars of order 6 are all constant: 
 
\begin{equation}\label{results2} 
(\nabla{C})^2 = \mu_1, ~ C_{1}^3 = \mu_2, ~ C_{2}^3 = \mu_3 . 
\end{equation} 
Since all of the basis scalars of order 
six are constant, then all scalars of order 
six are constant. 
 
We now proceed with the 
higher order scalars: orders (8,10,12) were considered in \cite{FKWC1992}. 
In particular, there is a sub-basis of scalar (rank-0) order 8 polynomials 
consisting of 92 elements given in Appendix B of \cite{FKWC1992} from which, by variation, 
we can obtain a set of independent conserved rank-2 tensors of order 8. 
For an Einstein space satisfying 
(\ref{einstein}), (\ref{C1})-(\ref{C2}) and (\ref{results2}), there are only 11 
independent scalars:  
$C^{pqrs;tu}C_{pqrs;tu}$ and 
$C^{pqrs}C_p^{~tuv} C_{qtru;sv}$,  
3 scalars (involving squares of the first covariant derivative) of the form   
$C^{prqs} C^{tuv}_{~~~~p;q} C_{tuvr;s}$, and 6 algebraic 
fourth order polynomials of the form 
$C^{pqrs} C_{pqr}^{~~~~t} C^{uvw}_{~~~~s} C _{uvwt}$. 
By obtaining the set of (more that 12) independent conserved rank-2 tensors of order 8, it 
follows that all of these 11 independent scalars are constant. In particular, 
the 3 scalars involving the first covariant derivative of the Weyl tensor 
are constant, and we are well on our way to showing that the spacetime is $CSI_1$. 
Indeed, in four-dimensions this is sufficient to show that the resulting 
spacetime is $CSI$ \cite{CSI,CSI4d}. 
Continuing in this way we obtain the result that 
in a universal spacetime all scalar curvature invariants
are constant.

An alternative proof, at least in a restricted range of applicability, is 
provided by the slice theorem.

\subsection{The slice theorem}

Let $I_i$ denote all possible polynomial scalar curvature invariants.  Then we can generate
a corresponding set of conserved symmetric tensors, $T_{i,\mu\nu}$, by
considering the variation of $S[I_i]=\int I_i\sqrt{-g}d^Nx$.

Let us assume that the spacetime under consideration is universal. 
If the spacetime is strongly universal then all of these symmetric tensors are zero: 
$T_{i,\mu\nu}=0$. If there is a $T_{i,\mu\nu}$ which is non-zero, 
then the spacetime is weakly universal, and, assuming that 
$T_{1,\mu\nu}=\lambda_1g_{\mu\nu}\neq 0$, we can then define the equivalent set of invariants: 
\[ \tilde{I}_1=I_1+2\lambda_1, \quad \tilde{I}_i=I_i-\frac{\lambda_i}{\lambda_1}I_1.\] 
We notice that for this new set of invariants, the corresponding conserved tensors are all zero: $\tilde{T}_{i,\mu\nu}=0$. 

This means that we have a full set of invariants each of which has a
zero variation:  $\frac{\delta S}{\delta g_{\mu\nu}}=0$.  This is a
signal that a universal metric has a degeneracy in its curvature
structure.  In particular, consider a metric variation $\delta
g_{\mu\nu}=\epsilon h_{\mu\nu}$, where $h_{\mu\nu}g^{\mu\nu}=0$
(traceless).  This implies that the variation with respect to
the metric is zero, which then  implies that the variation of all the invariants
in the direction of $h_{\mu\nu}$ is zero.  The metric is thus a
fixed point of all possible actions.

For the degenerate Kundt metrics there exists a one-parameter family
of metrics $g_\tau$ such that $I_i[g]=I_i[g_\tau]$.  Clearly, this
implies that $\lim_{\tau\rightarrow 0}\frac{S[g]-S[g_\tau]}{\tau}=0$;
i.e., $\frac{\delta S}{\delta g_{\mu\nu}}$ valishes along
$h_{\mu\nu}\equiv\lim_{\tau\rightarrow
0}(g_{\mu\nu}-g_{\tau,\mu\nu})/\tau$.  We note that for the Kundt
spacetimes this metric deformation can always be chosen to be
traceless (indeed, nilpotent).  The universality condition leads to
additional conditions since \emph{all} variations of the metric is
required to be zero.  However, degenerate Kundt metrics are
particularly promising candidates for universal metrics \cite{inv}.

In the Riemannian case, the slice theorem was used by Bleecker
\cite{Bleecker} to prove many results regarding critical metrics.  The
slice theorem considers the manifold of metrics modulo the
diffeomorphism group.  Ebin \cite{Ebin} proved the slice theorem for
the compact Riemannian case.  The Lorentzian case is  more
problematic and its general validity is questionable, but
Isenberg and Marsden \cite{IM} showed a slice theorem for solutions
to the Einstein equations given some assumptions (essentially, global
hyperbolicity and compact spatial sections).  In its infinitesimal
version, it states that any symmetric tensor can be split as follows:
\beq
S_{\mu\nu}=\pounds_{X}g_{\mu\nu}+T_{\mu\nu},
\eeq
for some vector field, $X$, and where $T$ is conserved: $\nabla^\mu T_{\nu\mu}=0$. 
The vector field $X$ can be interpreted as the generator of the diffeomorphism group 
and thus corresponds to a ``gauge freedom''. 

\emph{Consider the Lorentzian case when the slice theorem is valid}.  It can
now be shown that universality implies CSI (following parts of Bleecker's
argument).  Assume therefore that the spacetime is \emph{not CSI}.
Then there must exists a non-constant invariant $I$.  In particular,
there must exist a non-trivial interval $[a,b]$ onto which the
invariant $I$ is onto.  Therefore, choose a sufficiently small interval and a
smooth function $f(I)$.  The space of such functions is clearly
infinite dimensional.  Construct then the tensor deformation
$\tilde{g}_{\mu\nu}=(1+f(I))g_{\mu\nu}$.  By the slice theorem, there
exists a diffeomorphism $\phi$ such that $\phi^*\tilde{g}_{\mu\nu}$
is conserved.  Clearly, $\phi^*\tilde{g}_{\mu\nu}$ is an invariant
tensor and thus, by universality, $\phi^*\tilde{g}_{\mu\nu}=\lambda
g_{\mu\nu}$.  This implies that the metric deformation is a conformal
transformation.  However, the space of conformal transformations is
finite, thus, it must be possible to choose a $f(I)$ such that
$\phi^*\tilde{g}_{\mu\nu}\neq \lambda g_{\mu\nu}$.  Consequently, the
space is not universal.  To summarise, if a spacetime is not CSI, then
it is not universal.
Therefore, \emph{universality implies CSI}.

Note that in the compact Riemannian case the slice theorem holds and thus
universality implies CSI.  In the Riemannian case CSI implies local
homogeneity and thus this provides us with a slightly different proof to
that of Bleecker \cite{Bleecker}.

Of course, this result depends crucially on the range of applicability 
of the slice theorem. It is consequently  of interest to determine
for which Lorentzian spaces the slice theorem is valid.  
However, the result is important in the context here, since it can be seen that
there is a clear link between universality and CSI spaces, which lends further 
support to the  conjecture.

Let us next consider some properties of the Kundt-CSI spacetimes.

\section{Kundt CSI metrics}\label{sect:Kundt}

In \cite{CSI4d} it was proven that if a 4D spacetime is $CSI$, then either  
the spacetime is locally homogeneous  
or the spacetime is a degenerate Kundt spacetime.  
The Kundt-$CSI$ spacetimes are of particular interest  
since they are solutions of supergravity or superstring theory when supported by  
appropriate bosonic fields \cite{CFH}.  
It is plausible that a wide  
class of $CSI$ solutions are exact solutions to string theory  
non-perturbatively \cite{string}.  
In the context of string theory, it is of considerable interest to  
study   
higher dimensional Lorentzian $CSI$ spacetimes. In particular,   
a number of higher-dimensional $CSI$ spacetimes are also known to be  
solutions of supergravity theory \cite{CFH}.  
The supersymmetric properties of $CSI$ spacetimes have also  
been studied, particularly those that admit a null 
covariantly constant  
vector (CCNV) \cite{McNutt}.

A Kundt-$CSI$ can be written in the form   \cite{coley} 
\beq \d s^2=2\d u\left[\d v +H(v,u,x^k)\d   
u+W_{i}(v,u,x^k)\d x^i\right]+g^{\perp}_{ij}(x^k)\d x^i\d x^j,   
\label{HKundt}   
\eeq    
where the metric functions $H$ and $W_ {i}$, requiring $CSI_0$,  are given by   
\beq   
W_{i}(v,u,x^k)&=& v{W}_{i}^{(1)}(u,x^k)+{W}_{i}^{(0)}(u,x^k),\label{HKa}\\   
H(v,u,x^k)&=& {v^2}\tilde{\sigma}+v{H}^{(1)}(u,x^k)+{H}^{(0)}(u,x^k), \label{HKb} \\   
\tilde{\sigma} &\equiv& \frac 18\left(4\sigma+W^{(1)i}W^{(1)}_i\right),  
\label{sigma}    
\eeq   
where $\sigma$ is a constant.  
The remaining equations for $CSI_0$ that need to be solved are (hatted indices refer to an orthonormal frame in the transverse space):    
\beq   
 \label{Wcsi1}   
W^{(1)}_{[\hat i;\hat j]} &=& {\sf a}_{\hat i\hat j}, \\   
W^{(1)}_{(\hat i;\hat j)}-\frac 12 \left(W^{(1)}_{\hat i}\right)\left(W^{(1)}_{\hat j}\right) &=& {\sf s}_{\hat i\hat j},   
\label{Wcsi4}\eeq 
where the  ${\sf a}_{\hat i\hat j}$ and ${\sf s}_{\hat i\hat j}$  and
the components ${R}^{\perp}_{\hat i\hat j\hat m\hat n}$ are all  
constants (i.e., $\d S^2_H=g^{\perp}_{ij}(x^k)\d x^i\d x^j$ 
is curvature homogeneous). 
In four dimensions, $g^{\perp}_{ij}(x^k)\d x^i\d x^j$ is 2 dimensional, which immediately implies  $g^{\perp}_{ij}(x^k)\d x^i\d x^j$ is a 2 dimensional locally homogeneous space and, in fact, maximally symmetric space.  Up to scaling, there are (locally) only 3 such, namely the sphere, $S^2$; the flat plane, $\mathbb{E}^2$; and the hyperbolic plane, $\mathbb{H}^2$.   
   
The equations (\ref{Wcsi1}) and (\ref{Wcsi4}) now give a  
set of differential equations for $W^{(1)}_{\hat i}$.  
These equations determine uniquely $W^{(1)}_{\hat i}$ up to initial conditions  
(which may be free functions in $u$).  Also, requiring  
$CSI_1$ gives an additional set of constraints:    
\beq   
{\mbold\alpha}_{\hat i;}&=&\sigma W^{(1)}_{\hat i}-\frac 12({\sf s}_{\hat j\hat i}+{\sf a}_{\hat j\hat i})W^{(1)\hat j}, \\   
{\mbold\beta}_{\hat i\hat j\hat k}&=&W^{(1)\hat n}{R}^{\perp}_{\hat n\hat i\hat j\hat k}-W^{(1)}_{\hat i}{\sf a}_{\hat j\hat k}+({\sf s}_{\hat i[\hat j}+{\sf a}_{\hat i[\hat j})W^{(1)}_{\hat k]},   
\eeq   
where ${\mbold\alpha}_i$ and ${\mbold\beta}_{\hat i\hat j\hat k}$ are contants determined from the curvature invariants.    
We note that 
for a four-dimensional Kundt spacetime, $CSI_1$ implies $CSI$  \cite{CSI4d}.

There is a strong relationship between CSI spacetimes and those that are universal in four dimensions: 
\begin{thm}
A 4D universal spacetime of Petrov type {\bf {D}}, {\bf {II}}, or {\bf {III}}, is a Kundt-CSI spacetime.
\end{thm}
\begin{proof}

Consider type {\bf {D}} first.  Assuming that the spacetime is Einstein, 
then that spacetime is necessarily CSI$_0$ (which follows from
the previous discussion).  This implies that the boost weight  0 components are
constants in the canonical frame.  Using the Bianchi identities,
it immediately follows that it is also Kundt.  Since the previous
analysis also implies that it is CSI$_1$, then we have that the
spacetime is Kundt-CSI.

For type {\bf {II}}, the analysis is almost identical to the type {\bf {D}} analysis. For type {\bf {III}}, 
it is necessary to calculate some conserved tensors. Using the Weyl type {\bf {III}} canonical form, 
the Bianchi identities imply that the
Newman-Penrose spin coefficients \cite{NP} 
$\kappa=\sigma=0$ (and $\rho=\epsilon$, $\beta=\tau$, $\alpha=-2\pi$, $\gamma=-2\mu$). Requiring also that $H^{(6,3)(8)}_{ab}=\lambda g_{ab}$, say, gives the additional equation: 
$\rho^2=0$. Clearly, $\rho=0$ and the spacetime is Kundt. Since CSI$_1$ implies CSI for Kundt spacetimes, the theorem follows.  
\end{proof}

Although the theorem does not explicitly include Weyl type {\bf {N}} and {\bf {I}} spacetimes, it
is believed that these are also Kundt.  For type {\bf {I}} the expressions
for the conserved tensors are messy and unmanageable, and for type
{\bf {N}} it is necessary to compute a particular order 16 conserved tensor.

This proves the first two statements in the Conjecture \ref{con},
at least for Petrov types {\bf {D}}, {\bf {II}} and {\bf {III}} in 4 dimensions. However, we can 
see that the last two statements are also true: 
\begin{prop}
Consider a 4D Kundt CSI spacetime $(M,g)$. If $(M,g)$ is universal, then Conjecture \ref{con} is true. 
\end{prop}
\begin{proof}
The proof utilises the results from \cite{CSI4d}.  Assuming  an Einstein,
Kundt-CSI spacetime, the cases  reduces to those  where the corresponding
homogeneous spacetime $(\tilde{M},\tilde{g})$, is locally one of the
following:  Minkowski, de Sitter, anti-de Sitter, $dS_2\times S^2$,
or $AdS_2\times H^2$.  These have the corresponding Riemannian
counterparts, $(\hat{M},\hat{g})$, with identical invariants:  flat
space, $S^4$, $H^4$, $S^2\times S^2$, $H^2\times H^2$.
\end{proof}

We note that the opposite does not follow; namely, it is not true that for every
Riemannian universal spacetime there is a Lorentzian spacetime with the same
invariants.  For example, the symmetric spaces $CP^2$ and $H_{\mathbb{C}}^2$
(with the corresponding Fubini-Study and Bargmann metrics, respectively) do not
have Lorentzian counterparts.  Thus the Conjecture \ref{con} is signature
dependent.

Finally, let us comment on the situation in higher dimensions,
where much less is known. We have not yet proven that higher dimensional CSI spacetimes
are either locally homogeneous or degenerate  Kundt, which is necessary for a proof of
a higher dimensional version of Proposition 3.2, although we have conjectured that this is so \cite{CSI,CSI4d}.
It is also very likely (from the study of curvature operators) that  universality implies that a spacetime is degenerate Kundt,
but this has not been discussed explicitly to date. However, 
proving the higher dimensional version of the last part of Conjecture 1.1 will likely be more difficult since 
we first need the analytic extension
of the Kundt spacetime.

\newpage

{\em Acknowledgements}. We would like to thank Gary W. Gibbons for helpful comments on the current manuscript.  
This work was supported, in part, by NSERC of Canada.

\end{document}